\documentclass[smallextended,referee,envcountsect,]{svjour3}
\usepackage[utf8]{inputenc}
\smartqed
\usepackage{graphicx}
\journalname{JOTA}

\usepackage{amssymb}
\usepackage{amsmath}
\usepackage{mathdots} 
\usepackage{graphicx}
\usepackage{multirow}
\usepackage{nicematrix}
\usepackage{lipsum}
\usepackage{amsfonts}
\usepackage{epstopdf}
\usepackage{algorithmic}
\usepackage{algorithm}
\usepackage{amsopn}
\usepackage{blkarray}

\usepackage[shortlabels]{enumitem}
\usepackage{xargs}
\usepackage{mathtools}
\usepackage{subcaption, url}

\usepackage{verbatim}

\newcommand{\R}{\mathbb{R}}

\newcommand{\mbf}[1]{\mathbf{#1}}
\newcommand{\re}{\textup{Re}}
\newcommand{\im}{\textup{Im}}
\newcommandx{\vecS}[2][2=]{\mathbf{\mathcal{S}}^{\mathbf{#1}}_{#2}}
\newcommandx{\uvecS}[2][2=]{\underline{\mathbf{\mathcal{S}}}^{\mathbf{#1}}_{#2}}
\newcommandx{\vecW}[1][1=]{\mathbf{\mathcal{W}}_{#1}}
\newcommandx{\uvecW}[1][1=]{\underline{\mathbf{\mathcal{W}}}_{#1}}
\newcommandx{\vecA}[1][1=]{\mathbf{\mathcal{A}}_{#1}}
\newcommandx{\uvecA}[1][1=]{\underline{\mathbf{\mathcal{A}}}_{#1}}

\begin{document}

\title{Sparse optimization for estimating the cross-power spectrum in linear inverse models : from theory to the application in brain connectivity}
\titlerunning{Brain cross-power spectrum sparse optimization}

%[Brain cross-power spectrum sparse optimization]

\author{Laura Carini, Isabella Furci, Sara Sommariva}

\institute{Laura Carini \at
              University of Genoa, Geona, Italy \\
              laura.carini@edu.unige.it
           \and
           Isabella Furci \at
              University of Genoa, Geona, Italy \\
              isabella.furci@unige.it
            \and
              Sara Sommariva,  Corresponding author  \at
              University of Genoa, Geona, Italy \\
              sara.sommariva@unige.it
}

\date{Received: date / Accepted: date}
%The correct dates will be entered by the editor.

\maketitle

\begin{abstract}
In this work we present a computationally efficient linear optimization approach for estimating the cross--power spectrum of an hidden multivariate stochastic process from that of another observed process. Sparsity in the resulting estimator of the cross--power is induced through $\ell_1$ regularization and the Fast Iterative Shrinkage-Thresholding Algorithm (FISTA) is used for computing such an estimator. With respect to a standard implementation, we prove that a proper initialization step is sufficient to guarantee the required symmetric and antisymmetric properties of the involved quantities. Further, we show how structural properties of the forward operator can be exploited within the FISTA update in order to make our approach adequate also for large--scale problems such as those arising in context of brain functional connectivity.

The effectiveness of the proposed approach is shown in a practical scenario where we aim at quantifying the statistical relationships between brain regions in the context of non-invasive electromagnetic field recordings. Our results show that our method provide results with an higher specificity that classical approaches based on a two--step procedure where first the hidden process describing the brain activity is estimated through a linear optimization step and then the cortical cross--power spectrum is computed from the estimated time--series.
%Traditional approaches typically rely on a two-step process: first, estimating brain activity, followed by the computation of connectivity metrics, such as the cross-power spectrum. This approach often leads to suboptimal results due to parameter dependencies in the reconstruction phase. 

%To address these limitations, we introduce a novel one-step method that estimates the cross-power spectrum directly from the observable process, eliminating the need for a separate optimization step. The method leverages the structural properties of the forward matrix and features of the shrinkage operator of the FISTA algorithm, preserving the symmetric and antisymmetric properties of the involved quantities.

%Furthermore, numerical experiments demonstrate that our approach is computationally efficient and reduces the occurrence of spurious connections, thus improving the accuracy of brain connectivity estimation.

\end{abstract}
\keywords{Sparse Linear Optimization \and Large Scale Linear Systems \and Statistical multivariate analysis \and Stochastic Processes \and Functional Brain Connectivity}
\subclass{62H12 \and 65F50 \and	15A69 \and 65K10}

%All acknowledgements should be placed in the back of the paper after Conclusions..

\section{Introduction}
\label{sec:intro}
We consider the problem of computing sparse estimation of the complex-valued cross-power spectrum of a multivariate stochastic process $\left\{ \mathbf{X}(t) \right\}_{t \in \mathbb{R}}$, whose realizations, $\left\{ \mathbf{x}(t) \right\}_{t \in \mathbb{R}}$, can only be observed indirectly through the realization, $\left\{\mathbf{y}(t)\right\}_{t \in \mathbb{R}}$, of another observable stochastic process $\left\{ \mathbf{Y}(t) \right\}_{t \in \mathbb{R}}$. In particular, we show that if $\mathbf{Y}(t)$ is a linear combination of $\mathbf{X}(t)$ corrupted by additive noise,
then the problem can be formulated as a large-scale linear optimization problem \cite{vallarino2020}, whose dimension depends on the number of both the observations and the unknown variables.

Our work is motivated by the study of cortical functional connectivity from magneto-/electro-encephalographic (M/EEG) data, which consists in quantifying the statistical relationships between the activity of multiple brain regions from the electromagnetic field non-invasively recorded outside the scalp \cite{pereda_05,sakkalis_11}. From a mathematical point of view, this is typically achieved through a two-step procedure \cite{schoffelen2019}. First a time--series describing the activity of different brain regions is estimated from the recorded electromagnetic field by solving an ill-posed linear optimization problem \cite{baillet2001,ilsa19}. Then a proper connectivity metric is computed between the estimated time--series. Due to the oscillatory nature of the neural sources \cite{fr05}, many connectivity metrics are defined in the frequency domain, starting from the cross-power spectrum of the time--series describing the estimated brain activity \cite{nolte04,sommariva2019}.

A number of recent studies have demonstrated that this two--step procedure is inherently sub-optimal. When the initial optimization step is achieved through Tikhonov regularization \cite{hail94} a regularization parameter has to be set, and this is typically done in order to obtain the best possible estimate of the brain activity. However, the cross-power spectrum computed from the time--series estimated with this value of the regularization parameter is usually sub-optimal, in the sense that better estimate can be obtained from time--series computed using a different value of the regularization parameter \cite{vallarino2020,vallarino2023} depending on features of the neural time-courses themselves, such as their spectral complexity \cite{vallarino2021}. Similar results were obtained in a simulation setup where Tikhonov regularization was used to estimate the power spectrum of the neural sources, and connectivity was quantified by using coherence (a normalized version of the cross-power spectrum) \cite{hietal16}.

Furthermore, the estimates of the cortical cross-power spectrum obtained through this two--step procedure are usually affected by a large number of false positives, that is they identify statistical interactions between pairs of brain regions that are actually independent. One of the main causes of these spurious interactions is the residual mixing (or source--leakage) of the hidden process describing brain activity \cite{schoffelen2019}. To mitigate this issue, metrics insensitive to zero-lag interactions can be computed from the estimated cross--power spectrum, such as the imaginary part of coherency \cite{nolte04} or the weighted phase lag index \cite{stam2007,vinck2011}. However, such metrics have the obvious drawback that they also ignore instantaneous (or nearly instantaneous) true interactions \cite{ossadtchi2018}.

In the last years, few optimization techniques have been proposed to avoid such two--step procedure. 
In \cite{fukushima2015} the authors propose a new method of MEG source reconstruction that simultaneously estimates the source amplitudes and interactions across the whole brain by a variational Bayesian algorithm; in particular, they use a multivariate autoregressive (AR) model to represent directed interaction between sources together with a prior knowledge on structural brain connectivity inferred from diffusion magnetic resonance imaging (dMRI).
In \cite{tronarp2018} a Kalman filter is used to achieve a simultaneous estimation of source activities and their dynamic functional connectivity. However, also in this case AR models are used for representing source interactions. 
Ossadtchi and colleagues \cite{ossadtchi2018} introduce a novel projection matrix that operates on the cross--power spectrum of the observable process in order to mitigate the contribution of source--leakage to its real part. Once applied the projector, an optimization step still needs to be carried for estimating the cortical cross--power spectrum from the modified sensor--level data.

In this work we suggest an alternative approach where the cross--power spectrum of the hidden process describing the cortical activity is directly estimated from that of observed M/EEG time--series by setting up a linear optimization problem. In order to reduce the number of false positive, $\ell_1$ regularization \cite{figueiredo2007} is used so as to promote sparsity in the resulting estimator of the cortical cross--power spectrum. The Fast Iterative Shrinkage-Thresholding Algorithm (FISTA) \cite{beck2009} is a classical approach for solving linear optimization problems with $\ell_1$ penalty. However, its application to our problem is not trivial due to the high dimension of the matrix describing the forward operator which typically includes few hundreds of rows and few thousands of columns. Here we develop a computational strategy for efficiently carrying on the FISTA update by exploiting structural properties of the forward operator. The proposed method is validated on a large set of simulated data showing that our approach is capable of detecting the truly interacting sources but with an higher specificity than the classical two--step approach. The Python codes implementing the proposed approach are freely available at the
GitHub repository \url{https://github.com/theMIDAgroup/fista_cps_conn}.

The paper is organised as follows. In Section \ref{sec:problem} we derive the forward model relating the cross--power spectrum of the hidden process $\left\{ \mathbf{X}(t) \right\}_{t \in \mathbb{R}}$ to that of the observable one. In Section  \ref{sec:fista} we present the novel one--step approach for estimating the cross--power spectrum of  $\left\{ \mathbf{X}(t) \right\}_{t \in \mathbb{R}}$ and we recall a classical two-step approach used as benchmark, while Section \ref{sec:smart_computation} focuses on the strategy we develop to make our approach computationally affordable. Finally, numerical validations of the presented approach are shown in Section \ref{sec:num} and our conclusions are discussed in Section \ref{sec:conclusions.}.

\section{Problem formulation and forward modeling}
\label{sec:problem}
Throughout the paper we assume without loss of generality that the mean values of all the considered stochastic processes are zero. 

\begin{definition}\label{def:cps}
Let $\left\{ \mathbf{X}(t) \right\}_{t \in \mathbb{R}}$ and $\left\{ \mathbf{E}(t) \right\}_{t \in \mathbb{R}}$ be two real-valued, multivariate, stationary stochastic processes of dimension $n$ and $m$, respectively. Denoted with $\left\{\mathbf{\Gamma}^{\mathbf{X}\mathbf{E}}(\tau)\right\}_{t \in \mathbb{R}}$ the corresponding covariance function, that is $\mathbf{\Gamma}^{\mathbf{X}\mathbf{E}}(\tau) = \mathbb{E} \left[ \mathbf{X}(t)\mathbf{E}(t + \tau)^\top \right]$, we assume $\Gamma_{j,k}^{\mathbf{X}\mathbf{E}}(\tau)$ to be absolutely  integrable for all $j=1, \dots, n$ and  $k=1, \dots, m$. \\
Then, the cross-power spectrum between $\mathbf{X}(t)$ and $\mathbf{E}(t)$ is a one--parameter family of complex--valued matrices $\mathbf{S}^{\mathbf{X}\mathbf{E}}(f) \in \mathbb{C}^{n \times m}$ defined as \cite{bendat2011}
\begin{equation}
   S^{\mathbf{X}\mathbf{E}}(f) = \int_{-\infty}^{+\infty}  \Gamma^{\mathbf{X}\mathbf{E}}(\tau) e^{-2 \pi i f \tau} {\rm d} \tau,
\end{equation}
where the integral is computed component-wise.
\end{definition}

\begin{remark}
Following Definition \ref{def:cps}, we define the cross-power spectrum of the process $\{\mathbf{X}(t)\}_{t \in \mathbb{R}}$ as a one--parameter family of Hermitian matrices of size $n \times n$ denoted as $\mathbf{S}^{\mathbf{X}}(f) := \mathbf{S}^{\mathbf{X}\mathbf{X}}(f)$. Analogously, we denote $\mathbf{\Gamma}^{\mathbf{X}}(\tau) := \mathbf{\Gamma}^{\mathbf{X}\mathbf{X}}(\tau)$.
\end{remark}

\begin{theorem}\label{teo:fwd_model}
Let $\left\{ \mathbf{X}(t) \right\}_{t \in \mathbb{R}}$ and $\left\{ \mathbf{Y}(t) \right\}_{t \in \mathbb{R}}$ be two real-valued, multivariate, stationary stochastic processes of dimension $n$ and $m$, respectively. We further assume that $\mathbf{Y}(t)$ is a linear mixture of the components of $\mathbf{X}(t)$ corrupted by independent additive noise, that is
\begin{equation}\label{eq:fwd}
    \mathbf{Y}(t) = \mathbf{G}\mathbf{X}(t) + \mathbf{E}(t) \, , 
\end{equation}
where $\mathbf{G} \in \mathbb{R}^{m \times n}$ and $\left\{\mathbf{E}(t) \right\}_{t \in \mathbb{R}}$ is a multivariate, stationary stochastic process independent from $\left\{\mathbf{X}(t) \right\}_{t \in \mathbb{R}}$. Then, 
\begin{equation}\label{eq:fwd_matrix}
    \mathbf{S}^{\mathbf{Y}}(f) = \mathbf{G} \mathbf{S}^{\mathbf{X}}(f) \mathbf{G}^\top + \mathbf{S}^{\mathbf{E}}(f) \, .
\end{equation}
\end{theorem}

\proof{We observe that from the linearity of the expectation and the hypothesis on $\mathbf{Y}(t)$ it follows
\begin{displaymath}
\begin{split}
\mathbf{\Gamma}^{\mathbf{Y}}(\tau) & = \mathbb{E}\left[(\mathbf{G}\mathbf{X}(t)+\mathbf{E}(t)) (\mathbf{G}\mathbf{X}(t+\tau)+\mathbf{E}(t+\tau))^\top \right] \\
& = \mathbf{G} \mathbf{\Gamma}^{\mathbf{X}}(\tau) \mathbf{G}^\top + \mathbf{G} \mathbf{\Gamma}^{\mathbf{X}\mathbf{E}}(\tau) + \left( \mathbf{G} \mathbf{\Gamma}^{\mathbf{X}\mathbf{E}}(\tau)  \right)^\top +  \mathbf{\Gamma}^{\mathbf{E}}(\tau) \\
& = \mathbf{G} \mathbf{\Gamma}^{\mathbf{X}}(\tau) \mathbf{G}^\top +  \mathbf{\Gamma}^{\mathbf{E}}(\tau), \; 
\end{split}
\end{displaymath}
where the last equality holds due to the fact that $\left\{\mathbf{X}(t) \right\}_{t \in \mathbb{R}}$ and $\left\{\mathbf{E}(t) \right\}_{t \in \mathbb{R}}$ are independent and thus $\mathbf{\Gamma}^{\mathbf{X}\mathbf{E}}(\tau) = 0$ for all $\tau$. The thesis follows from Definition \ref{def:cps} by eploiting the linearity of the Fourier transform.\\}

Theorem \ref{teo:fwd_model} defines a linear relationship between the cross-power spectrum of the observable process $\{\mathbf{Y}(t)\}_{t \in \mathbb{R}}$ and that of the unknown process $\{\mathbf{X}(t)\}_{t \in \mathbb{R}}$. More explicitly, we denote with $\vecS{X}(f) \in \mathbb{C}^{n^2}$, and $\vecS{Y}(f), \vecS{E}(f) \in \mathbb{C}^{m^2}$ the vector obtained by stacking the columns ($\text{vec}(\cdot)$ operator) of matrices $\mathbf{S}^{\mathbf{X}}(f)$, $\mathbf{S}^{\mathbf{Y}}(f)$, and $\mathbf{S}^{\mathbf{E}}(f)$, respectively. From Eq. (\ref{eq:fwd_matrix}) it follows
\begin{equation}\label{eq:fwd_vec}
        \vecS{Y}(f) = \left( \mathbf{G} \otimes \mathbf{G} \right) \vecS{X}(f)  + \vecS{E}(f),
\end{equation}
where $\otimes$ is the Kronecker product.

By splitting real and imaginary parts, Eq. (\ref{eq:fwd_vec}) can be rewritten so to include only real-valued quantities:
\begin{equation}\label{eq:fwd_vec2}
\left( \begin{array}{c}
\re(\vecS{Y}(f)) \\
\im(\vecS{Y}(f)) \end{array} \right)
= \left( \begin{array}{cc}
\mbf{G} \otimes \mbf{G} & 0 \\
0 & \mbf{G} \otimes \mbf{G} \end{array} \right) 
\left( \begin{array}{c}
\re(\vecS{X}(f)) \\
\im(\vecS{X}(f)) \end{array} \right)
  + \left( \begin{array}{c}
\re(\vecS{E}(f)) \\
\im(\vecS{E}(f)) \end{array} \right) \; .
\end{equation}
We finally remark that, since $\mbf{S}^{\mbf{X}}(f)$ is an Hermitian matrix, $\re(\vecS{X}(f))$ and $\im(\vecS{X}(f))$ are the vectorization of a symmetric and antisymmetric matrix, respectively.  
%Finally, for the ease of notation, let us define
%\begin{align}
%\label{eq:cal_G}
%\uvecS{W}(f)\coloneqq&(\re\vecS{W}(f),\im\vecS{W}(f))^\top\\
%\mathcal{G}\coloneqq&\left( \begin{array}{cc}
%\mbf{G} \otimes \mbf{G} & 0 \\
%0 & \mbf{G} \otimes \mbf{G} \end{array} \right) 
%\end{align}
%so that equation \eqref{eq:fwd_vec2} can be written in compact form as
%\begin{equation}\label{eq:fwd_vec3}
%        \uvecS{Y}(f) =  \mathcal{G} \uvecS{X}(f)  + \uvecS{E}(f).
%\end{equation}

\section{Inverse modeling}
\label{sec:fista}

Given a realization $\{\mathbf{y}(t)\}_{t \in \mathbb{R}}$ of the observable process $\{\mathbf{Y}(t)\}_{t \in \mathbb{R}}$ our approach aims at providing a sparse estimation of the cross-power spectrum $\mathbf{S}^{\mathbf{X}}(f)$ by exploiting the model described by Eq. (\ref{eq:fwd_vec2}). When studying cortical functional connectivity from M/EEG data, this gives rise to a high-dimensional optimization problem. In fact the forward matrix 
\begin{equation}\label{eq:def_fwd_matrix}
\boldsymbol{\mathcal{G}} := \left( \begin{array}{cc}
\mbf{G} \otimes \mbf{G} & 0 \\
0 & \mbf{G} \otimes \mbf{G} \end{array} \right) \, 
\end{equation}
has size $2m^2 \times 2n^2$, where, in connectivity studies, $m \propto 10^2$ is the number of M/EEG sensors and $n \propto 10^3$ is the number of cortical locations where the brain activity and connectivity are estimated. Hence the matrix $\boldsymbol{\mathcal{G}} $ has size proportional to $(2 \cdot 10^4) \times (2 \cdot 10^6)$ making it necessary to develop proper computational strategies, described in Section \ref{sec:smart_computation}, to carry out any algorithmic step involving such a matrix.

Throughout the manuscript our approach is compared to a classical two-step approach summarized in the next subsection. 

%In this section we describe how we implemented FISTA for the specific linear model we are considering. Moreover, we introduce the classical two-step approach which we use to test the reliability and validity of our method.

\subsection{Benchmark: classical two--step approach}\label{sec:two_step}
A widely used approach for estimating the cross-power spectrum $\mathbf{S}^{\mathbf{X}}(f)$ consists in the following two steps \cite{schoffelen2019,vallarino2021,vallarino2023}.

\begin{description}
    \item [Step 1.] A regularised estimate, $\left\{\mbf{x}_\lambda(t)\right\}_{t \in \mathbb{R}}$, of the unknown process $\left\{\mbf{X}(t)\right\}_{t \in \mathbb{R}}$, is obtained by solving the inverse problem associated to equation \eqref{eq:fwd}. Here we consider the Tikhonov estimator \cite{tikhonov1963s}, which is defined as
    \begin{equation}\label{eq:mne_onestep}
        \mbf{x}_\lambda(t) = \arg \min_{\mbf{x}(t)}\left\{\|\mbf{G}\mbf{x}(t)-\mbf{y}(t)\|^2_2+\lambda\|\mbf{x}(t)\|^2_2\right\} \quad \forall t \in T \, , 
    \end{equation}
    where $\lambda$ is a proper regularisation parameter, $\|\cdot\|_2$ is the $\ell_2$-norm, and $T \subset \mathbb{N}$ is the discrete set of time points where the data were collected.
    \item[Step 2.] The Welch's estimator, $\mbf{S}^{\mbf{x}_\lambda}(f)$, of the time--series $\left\{\mbf{x}_\lambda(t)\right\}_{t \in T}$  is computed as follows \cite{welch1967}. First the time--series $\left\{\mbf{x}_\lambda(t)\right\}_{t \in T}$ is partitioned in $P$ overlapping segments of length $L$, denoted as $\left\{\mbf{x}^{(p)}_\lambda(\tau)\right\}_{\tau=0, \dots, L-1}$ with $p\in \left\{1, \dots, P\right\}$, and the discrete Fourier transform  
    \begin{equation}
    \mbf{\hat{x}}_\lambda^{(p)}(f) = \frac{1}{L}\sum_{\tau=0}^{L-1}\mbf{x}_\lambda^{(p)}(\tau)w(\tau)e^{\frac{-2\pi i \tau f}{L}}
    \end{equation}
    is computed for each segment, $\left\{ w(\tau)\right\}_{\tau=0, \dots, L-1}$ being the Hamming window \cite{hamming89}. Then we defined
    \begin{equation}
    \mbf{S}^{\mbf{x}_\lambda}(f) =\frac{L}{PW}\sum_{p=1}^P\hat{\mbf{x}}_\lambda^p(f) \hat{\mbf{x}}_\lambda^p(f)^H,
    \end{equation}
    where $W = \frac{1}{L}\sum_{t=0}^{L-1}w(t)^2$ .
\end{description}

%\textcolor{red}{Aggiungere una frase sul fatto che questo approccio produce molte connessioni spurie per giustificare che noi cerchiamo sparsità.}

\subsection{One--step approach for estimating the cross--power spectrum of the hidden process}
Fixed a frequency $f$, we propose to estimate the cross-power spectrum $\mathbf{S}^{\mathbf{X}}(f)$ through a least-squares approach with $\ell_1$ regularization applied to the model in Eq. (\ref{eq:fwd_vec2}). Hence, after computing the Welch's estimator $\mbf{S}^{\mbf{y}}(f)$ of the cross--power spectrum of the observed data,  we define \cite{tibshirani1996,figueiredo2007}

\begin{multline}\label{eq:l1_opt_pb}
 \left( \begin{array}{c}
\widehat{\re(\vecS{X})} \\
\widehat{\im(\vecS{X})} \end{array} \right)  =  \\ \underset{ \mbf{s}}{\mathrm{argmin}} \left\{ \Bigg|\Bigg| \left( \begin{array}{cc}
\mbf{G} \otimes \mbf{G} & 0 \\
0 & \mbf{G} \otimes \mbf{G} \end{array} \right) 
\mbf{s}- \left( \begin{array}{c}
\re(\vecS{y}) \\
\im(\vecS{y}) \end{array} \right) \Bigg|\Bigg|_2^2 + \lambda \|\mbf{s}\|_1\right\},
\end{multline}
where, for the sake of readability, we set $\mbf{s}=\begin{pmatrix}
     \mbf{s}_1\\
     \mbf{s}_2
 \end{pmatrix}$
 and we omit the specification of the argument $f$.

In order to solve the optimization problem in (\ref{eq:l1_opt_pb}) we used the Fast Iterative Shrinkage-Thresholding Algorithm (FISTA) \cite{beck2009} implemented as described in Algorithm \ref{alg:fista_one_step}. In particular we observe that in line 4 and 6 of the algorithm we have exploited the block-diagonal structure of the forward matrix defined in Eq. (\ref{eq:def_fwd_matrix}).

Following the original paper by Beck and Teboulle \cite{beck2009}, in Algorithm \ref{alg:fista_one_step} we used the shrinkage operator $\mathcal{T}_{\frac{\lambda}{L}} : \mathbb{R}^{n^2} \rightarrow \mathbb{R}^{n^2}$ defined as, for all $  i = 1, \dots, n^2$,

 \begin{equation}\label{eq:shrinkage_op}
\mathcal{T}_\frac{\lambda}{L}(\mbf{x})_i := \left(|x_i| - \frac{\lambda}{L} \right)^+ \text{sign}(x_i) = \begin{cases}
      x_i-\frac{\lambda}{L} & \text{if }x_i\geq\frac{\lambda}{L}\\
      0& \text{if } |x_i| \leq \frac{\lambda}{L}\\
      x_i+\frac{\lambda}{L} & \text{if }x_i\leq-\frac{\lambda}{L}
    \end{cases},
\end{equation}
where for all $z \in \mathbb{R}$, $(z)^+ := \max \left\{z, 0 \right\}$ and $\text{sign}(z)$ denote the ramp function (or positive part) and the sign function, respectively, and
\begin{equation}\label{eq:def_L}
L = 2\lambda_{\max}(\boldsymbol{\mathcal{G}} ^\top\boldsymbol{\mathcal{G}} )  \, , 
\end{equation}
being $\lambda_{\max}(\boldsymbol{\mathcal{G}} ^\top\boldsymbol{\mathcal{G}} )$ the highest eigenvalue of the squared matrix $\boldsymbol{\mathcal{G}} ^\top\boldsymbol{\mathcal{G}} $.

 \begin{remark} The constant $L$ in  Eq. (\ref{eq:def_L}) is the smallest Lipschitz constant of the gradient of a function \cite{beck2009}, that is, $\mbf{f}:\mathbb{R}^{2n^2} \rightarrow \mathbb{R}$ defined as $$\mbf{f}\begin{pmatrix}
     \mbf{s}_1\\
     \mbf{s}_2
 \end{pmatrix} = \left\|\boldsymbol{\mathcal{G}}\begin{pmatrix}
     \mbf{s}_1\\
     \mbf{s}_2
 \end{pmatrix} - \begin{pmatrix}
 \re(\vecS{y})\\ 
 \im(\vecS{y})\end{pmatrix}\right\|_2^2.$$
 Furthermore, by exploiting the block-diagonal structure of $\boldsymbol{\mathcal{G}}$ and the properties of the Kronecker product it holds
 \begin{equation}
 \begin{split}
 L  & = 2\lambda_{\max} \left( \left( \begin{array}{cc}
\mbf{G} \otimes \mbf{G} & 0 \\
0 & \mbf{G} \otimes \mbf{G} \end{array} \right)^\top \left( \begin{array}{cc}
\mbf{G} \otimes \mbf{G} & 0 \\
0 & \mbf{G} \otimes \mbf{G} \end{array} \right) \right) \\
& = 2\lambda_{\max}  \left( \begin{array}{cc}
\mbf{G}^\top \mbf{G} \otimes \mbf{G}^\top \mbf{G} & 0 \\
0 & \mbf{G}^\top \mbf{G} \otimes \mbf{G}^\top \mbf{G} \end{array} \right)  \\
& = 2 \left[\lambda_{\max} (\mbf{G}^\top \mbf{G})\right]^2.
 \end{split}
\end{equation}
 \end{remark}

\begin{algorithm}
\caption{FISTA for one--step estimation of the cross--power spectrum $\mbf{S}^{\mbf{x}}$}\label{alg:fista_one_step}
\renewcommand{\algorithmicrequire}{\textbf{Input:}}
\renewcommand{\algorithmicensure}{\textbf{Output:}}
\begin{algorithmic}[1]
\REQUIRE $\mbf{G} \in \mathbb{R}^{m \times n}$; $\lambda$, $L$, $K$, $\varepsilon > 0$; $\mbf{S}^{\mbf{y}} \in \mathbb{C}^{m \times m}$ and $\mbf{S}_0 \in \mathbb{C}^{n \times n}$ Hermitian matrices and corresponding vectorization  $\vecS{\mbf{y}} := \text{vec}(\mbf{S}^{\mbf{y}}) \in \mathbb{C}^{m^2}$ and $\vecS{}[0] := \text{vec}(\mbf{S}_0) \in \mathbb{C}^{n^2}$.
\vspace{3pt}
\STATE \textbf{Initialize:} $k=0$, $t_0 = 0$,  
$
\begin{pmatrix} \mbf{w}_{0,1} \\ \mbf{w}_{0,2} \end{pmatrix} = \begin{pmatrix} \mbf{s}_{0,1} \\ \mbf{s}_{0,2} \end{pmatrix} = \begin{pmatrix} \re(\vecS{}[0]) \\ \im(\vecS{}[0]) \end{pmatrix}
$
\vspace{3pt}
\WHILE{$k \leq K$ \AND $e \leq \varepsilon$}
\vspace{3pt}
    \STATE $k = k + 1$
    \vspace{3pt}
    \STATE 
    \vspace{3pt}
    $
    \begin{pmatrix} \mbf{s}_{k,1} \\ \mbf{s}_{k,2} \end{pmatrix} = 
    \begin{pmatrix}  
     \mathcal{T}_{\frac{\lambda}{L}} \left( \mbf{w}_{k-1,1} - \frac{2}{L} (\mbf{G} \otimes \mbf{G})^\top((\mbf{G} \otimes \mbf{G}) \mbf{w}_{k-1,1} - \re(\vecS{y})) \right) \\
     \mathcal{T}_{\frac{\lambda}{L}} \left( \mbf{w}_{k-1,2} - \frac{2}{L} (\mbf{G} \otimes \mbf{G})^\top((\mbf{G} \otimes \mbf{G}) \mbf{w}_{k-1,2} - \im(\vecS{y})) \right)
    \end{pmatrix}
    $
    \vspace{3pt}
    \STATE $t_k = \frac{1 + \sqrt{1+4t_{k-1}^2}}{2}$
    \vspace{3pt}
    \STATE
    \vspace{3pt}
    $
    \begin{pmatrix} \mbf{w}_{k,1} \\ \mbf{w}_{k,2} \end{pmatrix} = 
    \begin{pmatrix} \mbf{s}_{k,1} \\ \mbf{s}_{k,2} \end{pmatrix} + \frac{t_{k-1} - 1}{t_k} \begin{pmatrix} \mbf{s}_{k,1} - \mbf{s}_{k-1,1} \\ \mbf{s}_{k,2} - \mbf{s}_{k-1,2} \end{pmatrix}
$
\vspace{3pt}
    \STATE 
    \vspace{3pt}
    $
    e = \frac{||\mbf{s}_{k,1} - \mbf{s}_{k-1,1}||_1 + ||\mbf{s}_{k,2} - \mbf{s}_{k-1,2}||_1}{||\mbf{s}_{k,1}||_1 + ||\mbf{s}_{k,2}||_1}
    $
    \vspace{3pt}
\ENDWHILE
\vspace{3pt}
\ENSURE $\begin{pmatrix} \mbf{s}_{k,1} \\ \mbf{s}_{k,2} \end{pmatrix}$
\end{algorithmic}
\end{algorithm}

Since we seek a solution to the optimization problem (\ref{eq:l1_opt_pb}) to be interpreted as the cross-power spectrum of a stochastic process, $\mbf{s}_1$ and $\mbf{s}_2$ need to be the vectorization of a symmetric and antisymmetric matrix, respectively. The following results prove that this can be achieved by initializing FISTA through a random Hermitian matrix, as done in Algorithm \ref{alg:fista_one_step}.

\begin{proposition}\label{prop:shrinkage_operator}
Let $\vecA := {\rm vec}(\mbf{A})\in\mathbb{R}^{m^2}$ be the vectorization of a matrix $\mbf{A}\in \mathbb{R}^{m\times m}$ and let $\mathcal{T}_{\alpha}$ the shrinkage operator defined as in eq. (\ref{eq:shrinkage_op}) by replacing $\frac{\lambda}{L}$ with a generic parameter $\alpha>0$. The following properties hold:
\begin{enumerate}[(a)]
    \item if $\mbf{A}$ is symmetric, then $\mathcal{T}_\alpha(\vecA)$ is the vectorization of a symmetric matrix;
    \item if $\mbf{A}$ is antisymmetric, then $\mathcal{T}_\alpha(\vecA)$ is the vectorization of an antisymmetric matrix.
\end{enumerate}
\end{proposition}
\proof{Since $\mbf{A}_{ij} = \mathcal{A}_{m(j-1)+i}$  for all $i,j=1, \dots, m$, it can be shown that 
\[\mbf{A} \text{ is symmetric} \Longleftrightarrow \vecA[m(j-1)+i]=\vecA[m(i-1)+j], \, \text{for all} \, i,j = 1,\dots, m,\]
and 
\[\mbf{A} \text{ is antisymmetric} \Longleftrightarrow \vecA[m(j-1)+i]=-\vecA[m(i-1)+j],\, \text{for all}\, i,j = 1,\dots, m.\]
Hence, the thesis follows by the definition of $\mathcal{T}_\alpha$ observing that if $\mbf{A}$ is symmetric then 
\begin{displaymath}
\begin{split}
\mathcal{T}_\alpha(\vecA)_{m(j-1)+i} & = \left(|\mathcal{A}_{m(j-1)+i}| - \alpha \right)^+ \text{sign}(\mathcal{A}_{m(j-1)+i}) \\
& = \left(|\mathcal{A}_{m(i-1)+j}| - \alpha \right)^+ \text{sign}(\mathcal{A}_{m(i-1)+j}) = \mathcal{T}_\alpha(\vecA)_{m(i-1)+j}. 
\end{split}
\end{displaymath}
Similarly, if $\mbf{A}$ is antisymmetric, then 
\begin{displaymath}
\begin{split}
\mathcal{T}_\alpha(\vecA)_{m(j-1)+i} & = \left(|\mathcal{A}_{m(j-1)+i}| - \alpha \right)^+ \text{sign}(\mathcal{A}_{m(j-1)+i}) \\
& = \left(|-\mathcal{A}_{m(i-1)+j}| - \alpha \right)^+ \text{sign}(-\mathcal{A}_{m(i-1)+j}) = -\mathcal{T}_\alpha(\vecA)_{m(i-1)+j}. 
\end{split}
\end{displaymath}
\qed
}
\begin{theorem}
Let $\left\{    \begin{pmatrix} \mbf{s}_{k,1} \\ \mbf{s}_{k,2} \end{pmatrix}\right\}_{k\geq0}$ and $\left\{ \begin{pmatrix} \mbf{w}_{k,1} \\ \mbf{w}_{k,2} \end{pmatrix}\right\}_{k\geq0}$ the sequences generated with Algorithm \ref{alg:fista_one_step} with input the cross--power spectrum of the observed data $\mbf{S}^{\mbf{y}} \in \mathbb{C}^{m\times m}$ and a random Hermitian matrix $\mbf{S}_0 \in \mathbb{C}^{n \times n}$ as initial point. Then $\left\{\mbf{s}_{k,1}\right\}_{k\geq0}$ and $\left\{\mbf{w}_{k,1}\right\}_{k\geq0}$ are vectorization of symmetric matrices and $\left\{\mbf{s}_{k,2}\right\}_{k\geq0}$ and $\left\{\mbf{w}_{k,2}\right\}_{k\geq0}$ are vectorization of antisymmetric matrices.
\end{theorem}

\begin{proof}
To prove the theorem we proceed by induction. 

For $k=0$,  $
\begin{pmatrix} \mbf{w}_{0,1} \\ \mbf{w}_{0,2} \end{pmatrix} = \begin{pmatrix} \mbf{s}_{0,1} \\ \mbf{s}_{0,2} \end{pmatrix} = \begin{pmatrix} \re(\vecS{}[0]) \\ \im(\vecS{}[0]) \end{pmatrix}
$. Hence the thesis follows from the fact that $\mbf{S}_0$ is Hermitian.

Fixed $k>0$, we assume that $\mbf{s}_{k-1,1}$ and $\mbf{w}_{k-1,1}$ are vectorization of two symmetric matrices denoted with $\mbf{S}_{k-1,1}$ and $\mbf{W}_{k-1,1}$, respectively. Analogously, we assume that  $\mbf{s}_{k-1,2}$ and $\mbf{w}_{k-1,2}$ are vectorization of two antisymmetric matrices denoted with $\mbf{S}_{k-1,2}$ and $\mbf{W}_{k-1,2}$, respectively.

By exploiting the linearity of the $\text{vec}(\cdot)$ operator and the properties of the Kronecker product, from line 4 of Algorithm \ref{alg:fista_one_step} it follows:

\begin{equation}\label{eq:induction_1}
\begin{split}
    \mbf{s}_{k,1} & = \mathcal{T}_{\frac{\lambda}{L}} \left( \mbf{w}_{k-1,1}-\frac{2}{L} \left(\mbf{G}\otimes\mbf{G})^\top((\mbf{G}\otimes\mbf{G}) \mbf{w}_{k-1,1}-\re(\vecS{y})\right) \right) \\
    & = \mathcal{T}_{\frac{\lambda}{L}} \left( \text{vec}(\mbf{W}_{k-1,1})-\frac{2}{L} \left(\mbf{G}\otimes\mbf{G})^\top \text{vec}\left(\mbf{G} \mbf{W}_{k-1,1} \mbf{G}^\top-\re(\mbf{S}^{\mbf{y}}\right)\right) \right)\\
     & = \mathcal{T}_{\frac{\lambda}{L}} \left( \text{vec}\left(\mbf{W}_{k-1,1}-\frac{2}{L} \mbf{G}^\top \mbf{G} \mbf{W}_{k-1,1} \mbf{G}^\top \mbf{G} +\frac{2}{L} \mbf{G}^\top\re(\mbf{S}^{\mbf{y}})\mbf{G} \right)\right).
\end{split}
\end{equation}
Since $\mbf{W}_{k-1,1}$ and $\re(\mbf{S}^{\mbf{y}})$ are symmetric for the inductive hypothesis and the properties of the cross-power spectrum, respectively, the argument on the right side of the last equation in (\ref{eq:induction_1}) results to be a symmetric matrix. Hence Proposition \ref{prop:shrinkage_operator} ensures $ \mbf{s}_{k,1}$ is the vectorization of a symmetric matrix, $\mbf{S}_{k,1}$.

Similarly from line 4 of Algorithm \ref{alg:fista_one_step} it follows:

\begin{equation}\label{eq:induction_2}
\begin{split}
    \mbf{s}_{k,2} & = \mathcal{T}_{\frac{\lambda}{L}} \left( \mbf{w}_{k-1,2}-\frac{2}{L} (\mbf{G}\otimes\mbf{G})^\top((\mbf{G}\otimes\mbf{G}) \mbf{w}_{k-1,2}-\im(\vecS{y})) \right) \\
     & = \mathcal{T}_{\frac{\lambda}{L}} \left( \text{vec}\left(\mbf{W}_{k-1,2}-\frac{2}{L} \mbf{G}^\top \mbf{G} \mbf{W}_{k-1,2} \mbf{G}^\top \mbf{G} +\frac{2}{L} \mbf{G}^\top\im(\mbf{S}^{\mbf{y}})\mbf{G} \right)\right)
\end{split}
\end{equation}
where the argument of the right side of the last equation is antisymmetric because $\mbf{W}_{k-1,2}$ and $\im(\mbf{S}^{\mbf{y}})$ are antisymmetric. Hence Proposition \ref{prop:shrinkage_operator} ensures $ \mbf{s}_{k,2}$ is the vectorization of a antisymmetric matrix, $\mbf{S}_{k,2}$. 

The fact that $\mbf{w}_{k, 1}$ and $\mbf{w}_{k, 2}$ are the vectorization of a symmetric and an antisymmetric matrix then follows from line 6 of Algorithm \ref{alg:fista_one_step}, that can be rewritten as follows

\begin{equation*}
\left( \begin{array}{c}
\mbf{w}_{k,1} \\
\mbf{w}_{k,2} \end{array} \right) = \left( \begin{array}{c}
\text{vec} \left( \mbf{S}_{k,1}  + \frac{t_{k-1}-1}{t_k} (\mbf{S}_{k,1} -  \mbf{S}_{k-1,1}) \right) \\
\text{vec} \left( \mbf{S}_{k,2}  + \frac{t_{k-1}-1}{t_k} (\mbf{S}_{k,2} -  \mbf{S}_{k-1,2}) \right) \end{array} \right).
\end{equation*}

\end{proof}
\qed

\section{Smart product for an efficient computation of the FISTA update}
\label{sec:smart_computation}

The FISTA update at line 4 of Algorithm \ref{alg:fista_one_step} requires several matrix--vector multiplications involving the matrix $\mbf{G} \otimes \mbf{G}$ and its transpose  $(\mbf{G} \otimes \mbf{G})^\top 
= \mbf{G}^\top \otimes \mbf{G}^\top$. Since $\mbf{G} \in \mathbb{R}^{m \times n}$, $\mbf{G} \otimes \mbf{G}$ and $\mbf{G}^\top \otimes \mbf{G}^\top$ have size $m^2 \times n^2$ and $n^2 \times m^2$, respectively. Hence, the product of a vector by each one of these matrices has a cost proportional to $O(m^2n^2)$. In this section we show how the properties of the Kronecker product can be exploited to reduce both the computational cost and the memory requirements of such product. More specifically, our approach has two main advantages: on the one hand it reduces the computational cost to a value proportional to $O(\max(m,n)\, mn)$, on the other hand it avoids explicitly 
 assembling the matrix $\mbf{G} \otimes \mbf{G}$, by directly employing the matrix $\mbf{G}$.

First, we recall that the Kronecker product enjoys the  mixed-product property \cite{MR832183},  therefore it holds
\begin{equation}
\label{eq:split_tensor}
   \underset{m^2\times n^2 }{\underbrace {\mbf{G}\otimes \mbf{G}}}=
   (\mbf{G} \,\mbf{I}_{n}) \otimes (\mbf{I}_{m} \, \mbf{G}) = 
    \underset{m^2\times nm }{\underbrace { (\mbf{G}\otimes \mbf{I}_{m}) }}
  \, 
   \underset{nm\times n^2 }{\underbrace {  (\mbf{I}_{n}\otimes  \mbf{G})  }}.
\end{equation}

Now we present a result that establishes a connection between the elements of $\mbf{G} \otimes \mbf{I}_{m}$ and those of $\mbf{I}_{m} \otimes \mbf{G}$ through matrices that permute rows and columns. This relationship  plays a crucial role in the procedure for efficiently computing the matrix-vector product with $(\mbf{G} \otimes \mbf{G})$ (or its transpose, $(\mbf{G} \otimes \mbf{G})^\top$).

\begin{proposition}
    \label{prop:permutation}
   Let  $\mbf{G} \in \mathbb{R}^{m \times n}$ be an a rectangular matrix. Then there exist two matrices $\mbf{P}_{m^2}\in \mathbb{R}^{m^2 \times m^2} $ and $\mbf{P}_{nm}\in \mathbb{R}^{nm \times nm}$ such that
    \begin{equation}
        \label{eq:perm}
        \mbf{G} \otimes \mbf{I}_{m} = \mbf{P}_{m^2} ( \mbf{I}_{m}\otimes \mbf{G})\mbf{P}_{nm}.
    \end{equation}

\end{proposition}

\begin{proof}
The thesis is equivalent to identify two permutation functions for the row and column indices that rearrange the elements of $\mbf{I}_{m} \otimes \mbf{G}$, making it identical to $\mbf{G} \otimes \mbf{I}_{m}$. 
Relation (\ref{eq:perm}) is then established by defining $\mbf{P}_{m^2} \in \mathbb{R}^{m^2 \times m^2}$ and $\mbf{P}_{nm} \in \mathbb{R}^{nm \times nm}$ as identity matrices of the corresponding sizes, with the specified permutations applied to their rows and columns \cite{MR832183}. 

Then we first highlight that the  structures of  $\mbf{G}\otimes \mbf{I}_{m} $ and $\mbf{I}_{m} \otimes \mbf{G} $ are given by
\begin{equation*}
 \mbf{G}\otimes \mbf{I}_{m} =  \underset{nm  \text{ columns } }{ \underbrace{\begin{bmatrix}
g_{1,1}\mbf{I}_{m}  & g_{1,2}\mbf{I}_{m} & \cdots & g_{1,n}\mbf{I}_{m} \\
g_{2,1}\mbf{I}_{m} & g_{2,2}\mbf{I}_{m}  & \cdots & g_{2,n}\mbf{I}_{m} \\
\vdots & \vdots & \ddots & \vdots \\
g_{m,1}\mbf{I}_{m} & g_{m,2}\mbf{I}_{m}  & \cdots & g_{m,n}\mbf{I}_{m}  \\
\end{bmatrix}}},
\quad  \mbf{I}_{m}  \otimes \mbf{G} = \left. \begin{bmatrix}
\mbf{G} & O_{m,n} & \cdots &  O_{m,n} \\
 O_{m,n} & \mbf{G} & \cdots &  O_{m,n} \\
\vdots & \vdots & \ddots & \vdots \\
 O_{m,n} &  O_{m,n} & \cdots & \mbf{G}\\
\end{bmatrix}.\right\}m^2  \text{ rows }
    \end{equation*}
Consequently, we obtain the elements of $\mbf{G}\otimes \mbf{I}_{m}$ from those of $ \mbf{I}_{m} \otimes \mbf{G} $ with a permutation that brings the $i-$th row in position 
\begin{equation}
\label{eq:perm_rows}
    {\rm mod}(i-1,m) \, m+  \left\lfloor  \frac{i-1}{m} \right\rfloor +1, \quad {\rm for }\, i = 1\dots m^2;
\end{equation}
and the $j-$th column in position 
\begin{equation}
\label{eq:perm_colo}
    {\rm mod}(j-1,n) \, m+ \left\lfloor  \frac{j-1}{n} \right\rfloor +1,  \quad {\rm for }\, j = 1\dots nm,
\end{equation}
where, for all $a, b \in \mathbb{N}$ with ${\rm mod}(a,b)$ we denote the reminder after the division of $a$ by $b$ and with $\left\lfloor  a \right\rfloor$ the floor of $a$.

Indeed, note that, ${\rm mod}(i-1, m)$ (resp. ${\rm mod}(j-1, n)$ ) is the quantity that serves to group the row indices (resp. coloumn) into blocks of $m$ rows (resp. coloumn), determining the relative position within each block. The value $\left\lfloor \frac{i-1}{m} \right\rfloor$ (resp.  $\left\lfloor \frac{j-1}{n} \right\rfloor$ ) determines the block number we are considering.  See \cite[Section 2.5]{MR4449208} and \cite[Remarks 2-3]{MR4623368} for further visualizations and definitions on the block rectangular permutations.

Let $\mbf{e}_\ell$ be the $\ell-$th canonical vector with 1 in position $\ell$ and zeros elsewhere, that is $(\mbf{e}_\ell)_j = \delta_{j \ell}$. We conclude the proof choosing $ \mbf{P}_{m^2}$ and $ \mbf{P}_{nm}$ the two permutation matrices where the 
one entry equal to 1 is given exploiting the two permutation functions (\ref{eq:perm_rows}) and (\ref{eq:perm_colo}). In details,
\begin{equation*}
     \mbf{P}_{m^2} =\begin{bNiceMatrix}[margin]
    e_{\xi_r(1)} & | & e_{\xi_r(2)}  & | & \cdots & | & e_{\xi_r(
    m^2)}  \\
\end{bNiceMatrix}, \quad 
 \mbf{P}_{nm}=\begin{bNiceMatrix}[margin]
    e_{\xi_c(1)} & | & e_{\xi_c(2)}  & | & \cdots & | & e_{\xi_c(nm)}  \\
\end{bNiceMatrix},
\end{equation*}
where 
\begin{equation}
\label{eq:function_perm}
\begin{split}
    \xi_r(i)&=  {\rm mod}(i-1,m) \, m+  \left\lfloor  \frac{i-1}{m} \right\rfloor +1, \quad i = 1\dots m^2, \\
    \xi_c(j)&=  {\rm mod}(j-1,n) \, m+  \left\lfloor  \frac{j-1}{n} \right\rfloor +1, \quad j = 1\dots nm.
    \end{split}
\end{equation}

   \end{proof}   
\qed

\begin{theorem}
    \label{prop:smart_prod}
     Let  $\mbf{G} \in \mathbb{R}^{m \times n}$ be a rectangular matrix and $\mbf{x}\in \R^{n^2}$. The matrix--vector product $(\mbf{G}\otimes \mbf{G})\mbf{x}$ can be performed in $O(\max(m,n) \, mn)$ operations with a matrix--less approach (denoted by Algorithm \ref{alg:smart_prod}).
\end{theorem}

\begin{proof}
Equation (\ref{eq:split_tensor}) and Proposition \ref{prop:permutation} imply
\begin{displaymath}
(\mbf{G}\otimes \mbf{G})\mbf{x} = (\mbf{G} \otimes \mbf{I}_{m})(\mbf{I}_{n}\otimes  \mbf{G}) \mbf{x} = \mbf{P}_{m^2} ( \mbf{I}_{m}\otimes \mbf{G})\mbf{P}_{nm} (\mbf{I}_{n}\otimes  \mbf{G}) \mbf{x} \,,
\end{displaymath}
where $ \mbf{P}_{m^2}$ and $\mbf{P}_{nm}$ are the permutation matrices defined by functions $\xi_r$ and $\xi_c$ in Equation (\ref{eq:function_perm}) in  Proposition \ref{prop:permutation}.
Hence the product $(\mbf{G}\otimes \mbf{G})\mbf{x}$ can be decomposed in few steps that form Algorithm \ref{alg:smart_prod}.
We now analyze the proposed procedure and demonstrate that it does not require assembling any of the matrix tensor products. Moreover, we prove that the total computational cost of each step is proportional to $\max(m, n) \cdot mn$.

The cost of step 2 and 4 of Algorithm \ref{alg:smart_prod} is 0  since they only require  the index exchange suggested by  $\xi_r$ and $\xi_c$ in Equation (\ref{eq:function_perm}).
The vector $\mbf{x}\in \R^{n^2}$ can be split into $n$ vectors $\hat{x}_\ell$ of size $n$. In particular,\begin{equation*}
 \hat{x}_\ell= \begin{bmatrix}
     x_{(\ell-1)n+k}
 \end{bmatrix}_{k=1}^{n}, \quad \ell=1,\dots,n.
 \end{equation*}
Similarly, the vector $\widehat{\mbf{y}} $  can be split into $m$ vectors of size $n$.

The latter and the block diagonal structures of $\mbf{I}_n\otimes \mbf{G}$ and $\mbf{I}_m\otimes \mbf{G}$, imply that the multiplications in steps 1 and 3  of Algorithm \ref{alg:smart_prod}  avoid assembling the matrix $\mbf{G} \otimes \mbf{G}$ (or any other matrix tensor product). Moreover, they consist of $n$ and $m$ matrix-vector products with the matrix $\mbf{G}$, respectively. Therefore the total computational cost is proportional to $n\cdot nm+m\cdot n m$, that is $O(\max(n,m)\cdot mn)$.
\qed
\end{proof}

\begin{algorithm}
\caption{Efficient computation of the matrix--vector multiplication \\ $(\mbf{G}\otimes \mbf{G})\mbf{x}$}\label{alg:smart_prod}
\renewcommand{\algorithmicrequire}{\textbf{Input:}}
\renewcommand{\algorithmicensure}{\textbf{Output:}}
\begin{algorithmic}[1]
\REQUIRE $\mbf{G} \in \mathbb{R}^{m \times n}$; $\mbf{x}\in \R^{n^2}$; $\mbf{P}_{nm}$ and $\mbf{P}_{m^2}$ as in Proposition \ref{prop:permutation}.
\STATE{Compute $\mbf{y} := \left( \mbf{I}_n \otimes \mbf{G} \right) \mbf{x} \in \mathbb{R}^{nm}$;}
\STATE{Permute the element of $\mbf{y}$ to define $\widehat{\mbf{y}} := \mbf{P}_{nm} \mbf{y}$;}
\STATE{Compute $\mbf{z} := \left( \mbf{I}_m \otimes \mbf{G} \right) \widehat{\mbf{y}} \in \mathbb{R}^{m^2}$;}
\STATE{Permute the element of $\mbf{z}$ to define $\widehat{\mbf{z}} := \mbf{P}_{m^2} \mbf{z}$;}
\ENSURE{$\widehat{\mbf{z}}$}
\end{algorithmic}
\end{algorithm}

\begin{remark}
Since the results presented in this section hold true for any generic rectangular matrix, Algorithm \ref{alg:smart_prod} can be straightforwardly applied for computing also the matrix--vector product $\left(\mbf{G}^\top \otimes \mbf{G}^\top \right) \mbf{x}.$
Moreover,  in connectivity studies, $m \propto 10^2$ is the number of M/EEG sensors and $n \propto 10^3$ is the number of cortical locations and in general   $n<<m$. Then, the estimation cost of the Theorem \ref{prop:smart_prod} becomes $O( m^2n)$. 
\end{remark}

\section{Numerical validation}
\label{sec:num}

The numerical Section is organized as follows. In Subsection \ref{ssec:num_configuration} and \ref{ssec:simulations} we describe how MEG synthetic data were simulated an analyzed using both the classical two--step approach and the proposed one-step approach. The results and comparison of the two methods are summarized in Subsection \ref{ssec:results}.

\subsection{Brain connectivity configurations.}
\label{ssec:num_configuration}
To test the performance of the proposed approach under different experimental conditions, we considered the following two brain configurations.

\paragraph{Configuration 1: three active sources with unidirectional coupling from source 1 to source 2.}
Inspired by previous works \cite{sommariva2019,vallarino2023,haufe2019}, the 
time--courses of the active sources were simulated by filtering in the band $[ 8, 12]$Hz ($\alpha$ band) a multivariate autoregressive (MVAR)  process of order $P = 5$ defined as
\begin{equation}\label{eq:mvar_conf1}
    \begin{pmatrix}
        z_1(t) \\
        z_2(t) \\
        z_3(t)
    \end{pmatrix} = \sum_{k=1}^P
    \begin{pmatrix}
        a_{1,1}(k) & 0 & 0\\
        a_{2,1}(k) & a_{2,2}(k) & 0\\
        0 & 0 &a_{3,3}(k)
    \end{pmatrix}
        \begin{pmatrix}
        z_1(t-k) \\
        z_2(t-k) \\
        z_3(t-k)
    \end{pmatrix} + 
        \begin{pmatrix}
        \epsilon_1(t) \\
        \epsilon_2(t) \\
        \epsilon_3(t)
    \end{pmatrix} \, . 
\end{equation}
The non-zero elements $a_{i,j}(k)$ of the coefficient matrix were drawn from a normal distribution of zero mean and standard deviation $0.9$. We retained only coefficients resulting in $(i)$ a stable MVAR process \cite{lutkepohl2005new} and $(ii)$ triplets of signals, $(z_1(t), z_2(t), z_3(t))^T $,  such that the $\ell_2$ norm of the strongest one is less than 3 times the $\ell_2$ norm of the weakest one and such that the average over the range $[8, 12]$Hz of the sum of their power spectra was at least 1.2 times the average over the entire frequency range \cite{vallarino2023}.

\paragraph{Configuration 2: three sources with unidirectional coupling from source 1 to source 2 and 3.} The time--courses of the active sources were generating as for Configuration 1, but substituting the model in Eq. (\ref{eq:mvar_conf1}) with
    \begin{equation}\label{eq:mvar_conf2}
    \begin{pmatrix}
        z_1(t) \\
        z_2(t) \\
        z_3(t)
    \end{pmatrix} = \sum_{k=1}^P
    \begin{pmatrix}
        a_{1,1}(k) & 0 & 0\\
        a_{2,1}(k) & a_{2,2}(k) & 0\\
        a_{3,1}(k) & 0 &a_{3,3}(k)
    \end{pmatrix}
        \begin{pmatrix}
        z_1(t-k) \\
        z_2(t-k) \\
        z_3(t-k)
    \end{pmatrix} + 
        \begin{pmatrix}
        \epsilon_1(t) \\
        \epsilon_2(t) \\
        \epsilon_3(t)
    \end{pmatrix} \, .
\end{equation}

\subsection{Simulation and analysis of the observed MEG time--series.}\label{ssec:simulations}

We exploited the model in Eq. (\ref{eq:fwd}) for generating 50 realisations, $\{\mbf{y}(t)\}_{t=1}^T$, of the observable process for each configuration defined in the previous section. Specifically, we fixed $T=10,000$, and for all $t=1, \dots, T$, we set
\begin{equation}\label{eq:fwd_model_sim}
\mbf{y}(t) = \mbf{G} \mbf{x}(t) + \mbf{e}(t),
\end{equation}
where
\begin{itemize}
\item we extracted the forward operator $\mbf{G}$ from the sample dataset within the MNE Python package \cite{GramfortEtAl2013a} by considering only magnetometers, and by downsampling the available source space to $n=6940$ points. Hence, in our numerical experiment $\mbf{G}$ has size $102 \times 6940$ and each column $\mbf{g}_i$, $i=1, \dots, n$, represents the magnetic field generated by a point-like unit source placed at the $i$--th point of the source-space with orientation normal to the local cortical surface;
\item we defined $\mbf{x}(t)$ by randomly selecting three points of the source--space so that the pairwise source distances were greater than 4 cm and the pairwise ratios of the $\ell_2$ norms of the corresponding columns of $\mbf{G}$ were close to one. The components of $\mbf{x}(t)$ corresponding to the drawn location were set equal to the time--courses defined by equation (\ref{eq:mvar_conf1}) or (\ref{eq:mvar_conf2}) while the value of the remaining $n-3$ components was kept equal to 0;
\item we sampled $\mbf{e}(t)$ from a multivariate Gaussian distribution $\mathcal{N}(\mbf{0}, \sigma^2 \mbf{I}_m)$ where we chose $\sigma^2$ so as to obtain a signal-to-noise ratio (SNR) equal to 5 dB. 
\end{itemize}
To test the proposed method, for each one of the simulated data, $\{\mbf{y}(t)\}_{t=1}^T$, we then computed the Welch's estimator, $\mbf{S}^{\mbf{y}}(f)$, and we applied Algorithm \ref{alg:fista_one_step} by setting $f$ equal to the frequency in the range [8, 12]Hz where the component $S_{12}^{\mbf{y}}(f)$ peaks. In our experiments, we set the maximum number of iterations $K$ equal to 5,000, and the tolerance $\varepsilon$ equal to $10^{-5}$. To avoid inverse crime and mimic real-life scenarios where the active brain sources seldom match points of the source-space, the matrix $\mbf{G}$ used within Algorithm \ref{alg:fista_one_step} is obtained from that employed in Eq. (\ref{eq:fwd_model_sim}) for simulating the MEG data by further reducing the source-space to $n=644$ points. Finally, we tested four different values of the regularization parameters, by choosing four different scaling factor $\kappa$ evenly spaced in log-space in the range $[10^{-2}, 10^{-1}]$. Then we set $\lambda = \kappa \, \lambda^*$, where $\lambda^*= 2\, \left\|\boldsymbol{\mathcal{G}} 
\begin{pmatrix}
 \re(\vecS{y})\\ 
 \im(\vecS{y})\end{pmatrix}\right\|_{\infty}^2$, being $\vecS{y} = \text{vec}(\mbf{S}^{\mbf{y}}(f))$. The value $\lambda^*$ has been shown to be an upper bound for the optimization problem (\ref{eq:l1_opt_pb}) to admit a non-null solution \cite{gerstoft2015}. 

The performance of the proposed approach were compared to that of the classical two--step approach described in Section \ref{sec:two_step}. In detail, we first computed the Tikhonov estimator $\{\mbf{x}_\lambda(t)\}_{t=1}^T$ of the neural sources as in equation (\ref{eq:mne_onestep}) by using the coarse operator $\mbf{G}$ including $n=644$ source locations used also within Algorithm \ref{alg:fista_one_step}. Four different regularization parameters were tested, namely $\lambda = \xi \ 10^{-SNR/10}$, with $\xi \in \{0.1, 1, 10, 100\}$. For each value of the parameter, we then computed the Welch’s estimator, $\mbf{S}^{\mbf{x}_\lambda}(f)$, of the time--series $\left\{\mbf{x}_\lambda(t)\right\}_{t \in T}$.

Let $\widehat{\mbf{S}}$ be an estimate of the cross-power spectrum obtained with either the proposed method or the classical two--step approach. In order to quantitatively compare the two approaches, we separated the real and the imaginary part of $\widehat{\mbf{S}}$. For both of them, we then computed a weighted sum over all the estimated pairwise interaction exceeding a given threshold of the euclidean distance between the interacting-source locations and the closest pair of sources truly connected in sense of the Wasserstein 2-distance. More formally, denoted with $\mathcal{V} = \left\{ \mbf{v}_1, \dots, \mbf{v}_{6940} \right\}$ and $\mathcal{W} = \left\{ \mbf{w}_1, \dots, \mbf{w}_{644} \right\}$ the source spaces associated to the forward operators used for simulating the data and within the inverse procedures, respectively, we defined
\begin{equation}\label{eq:eval_crit_re}
\textrm{Err}^\re = \sum_{(\mbf{w}_i, \mbf{w}_j) \in \mathcal{E}^{rec}} \frac{\left|\re\left(\widehat{S}_{ij}\right)\right|}{\underset{i<j}{\max} \left|\re\left(\widehat{S}_{ij}\right)\right|} \underset{(\mbf{v}_p, \mbf{v}_q) \in \mathcal{E}^{true}}{\min} d\left((\mbf{w}_i, \mbf{w}_j), (\mbf{v}_p, \mbf{v}_q) \right),
\end{equation}
where $\mathcal{E}^{rec} = \left\{ (\mbf{w}_i, \mbf{w}_j) \in \mathcal{W} \times \mathcal{W}\ |\ i < j \ \text{and} \ \left|\re\left(\widehat{S}_{ij}\right)\right| \geq \tau \right\}$ collects the pairs of source location between which the estimated cross spectrum exceeds a given threshold, namely $\tau = 0.5 \ \underset{i<j}{\max}\left\{\left|\re\left(\widehat{S}_{ij}\right)\right| \right\}$;  $\mathcal{E}^{true} \subset \mathcal{V} \times \mathcal{V}$ is the set of truly connected sources, and $$d\left((\mbf{w}_i, \mbf{w}_j), (\mbf{v}_p, \mbf{v}_q) \right) = \sqrt{\frac{1}{2}\min\left\{ ||(\mbf{w}_i, \mbf{w}_j) - (\mbf{v}_p, \mbf{v}_q) ||_2^2 ,  ||(\mbf{w}_i, \mbf{w}_j) - (\mbf{v}_q, \mbf{v}_p) ||_2^2 \right\}}$$
is the Wasserstein 2-distance \cite{hoffman2004}. 

Similarly, $\textrm{Err}^\im$ was defined as in (\ref{eq:eval_crit_re}) by replacing $\re\left(\widehat{S}_{ij}\right)$ with $\im\left(\widehat{S}_{ij}\right)$.

\subsection{Results}
\label{ssec:results}
Table \ref{tab:res_sparsity} illustrated the behavior of the proposed one--step approach when varying the amount of regularization. As expected, the higher the value of the regularization parameter the sparser the resulting estimation of the cross--power spectrum. In detail, for the highest value of the parameter, namely $\lambda_4 := 0.1 \, \lambda^*$, for many simulated data the resulting estimation of the cross-power spectrum does not show any non-null interaction. On the other hand, for the lowest value of the parameter, namely $\lambda_1 := 0.01 \, \lambda^*$, only in two dataset for both the configurations the imaginary part of the estimated cross-power spectrum was equal to zero, but the number of supra-threshold connections increased up to few hundreds for Configuration 1 and few thousands for Configuration 2. We recall that the threshold $\tau$ was set equal to half the value of the strongest connection.

\begin{table}[h!]
 \centering 
 \caption{Level of sparsity in the solution provided by the proposed one-step approach for decreasing values of the regularization parameters. For each configuration, each cell of the table shows in the first row the percentage of simulated data where the real (first column) or the imaginary (second column) part of the estimated cross-power spectrum show at least one non-null interaction; in the second row the minimum and maximum number of supra-threshold connections across these data; and in the third row the mean number of supra-threshold interactions.}\label{tab:res_sparsity}
	\begin{tabular}{|c||*{2}{c}||*{2}{c}||}	
    \hline
  & \multicolumn{2}{c||}{\textbf{Conf 1}} & \multicolumn{2}{|c|}{\textbf{Conf 2}} \\
    \cline{2-5} 
    & \textbf{Real part}  & \textbf{Imag. part}   & \textbf{Real part}  & \textbf{Imag. part} \\
    \hline 
\multirow{3}{*}{$\lambda_4$} & 98.0\% & 58.0\% & 100.0\% & 78.0\% \\
& (1, 87) & (1, 41) & (1, 292) & (1, 284) \\
& 4 & 4& 17 & 17 \\
\hline 
\multirow{3}{*}{$\lambda_3$} & 98.0\% & 76.0\% & 100.0\% & 88.0\% \\
&(1, 156) & (1, 318) & (1, 401) & (1, 221) \\
& 6 & 18 & 19 & 14 \\
\hline 
\multirow{3}{*}{$\lambda_2$} &100.0\% & 86.0\% & 100.0\% & 92.0\% \\
&(1, 281) & (1, 462) & (1, 1013) & (1, 526) \\
& 20 & 31 & 44 & 30 \\
\hline 
\multirow{3}{*}{$\lambda_1$} &100.0\% & 96.0\% & 100.0\% & 96.0\% \\
&(1, 260) & (1, 325) & (1, 4303) & (1, 2942) \\
& 20 & 42 & 144 & 120 \\
 \hline  
 \end{tabular}	
  \end{table}

Fig. \ref{fig:sim26_config2} shows the main advantages of the proposed method over the classical two-step approach. In this example, both method identified connected sources nearby the truly interacting ones, however the number of false positive is much higher for the two--step approach. Further, for the two--step approach the value of the cross-power spectrum is underestimated of several orders of magnitude.

\begin{figure}[h!]
    \centering
    \includegraphics[width=1.\linewidth]{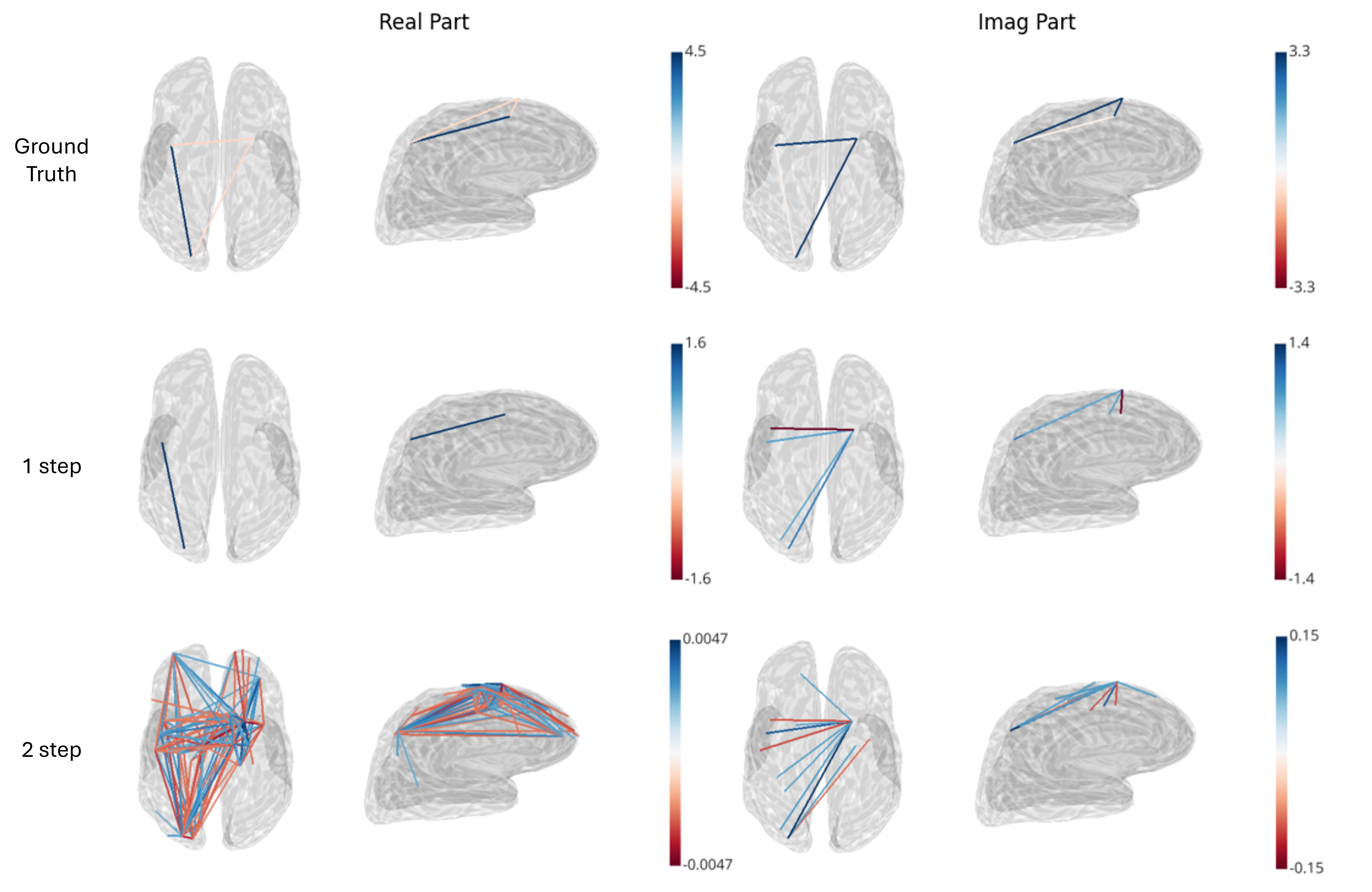}
    \caption{Original and estimated cross-power spectrum for one MEG data simulated by using Configuration 2. In this simulation $\textrm{Err}^\re$ is equal to 0.077 for the one--step approach and to 0.521 for the two--step, while $\textrm{Err}^\im$ is equal to 0.019 and 0.086 respectively. }
    \label{fig:sim26_config2}
\end{figure}

More in general, in our experiments, the one-step approach outperformed the classical two-step method as demonstrated by Fig. \ref{fig:quant_comp}. More specifically, for both the approaches we selected for each simulated data the best regularization parameter among those tested as the one minimizing the sum $\textrm{Err}^\re + \textrm{Err}^\im$. When using these optimal parameters the average error of the one--step approach is systematically lower than that of the two--step approach for both the real and the imaginary part of the cross--power spectrum in both the configurations. This is due to the fact that the number of spurious interactions is much higher for the two--step approach than for the proposed method, as shown in the second row of Fig. \ref{fig:quant_comp}.

\begin{figure}[h!]
\centering
    % \begin{subfigure}[b]{0.4\textwidth}
    % \includegraphics[width=\textwidth]{fig/ris_comp_1.png}
    % \end{subfigure}
    % \begin{subfigure}[b]{0.4\textwidth}
    % \centering
    % \includegraphics[width=\textwidth]{fig/ris_comp_2.png}
    % \end{subfigure}
   \includegraphics[width=0.7\linewidth]{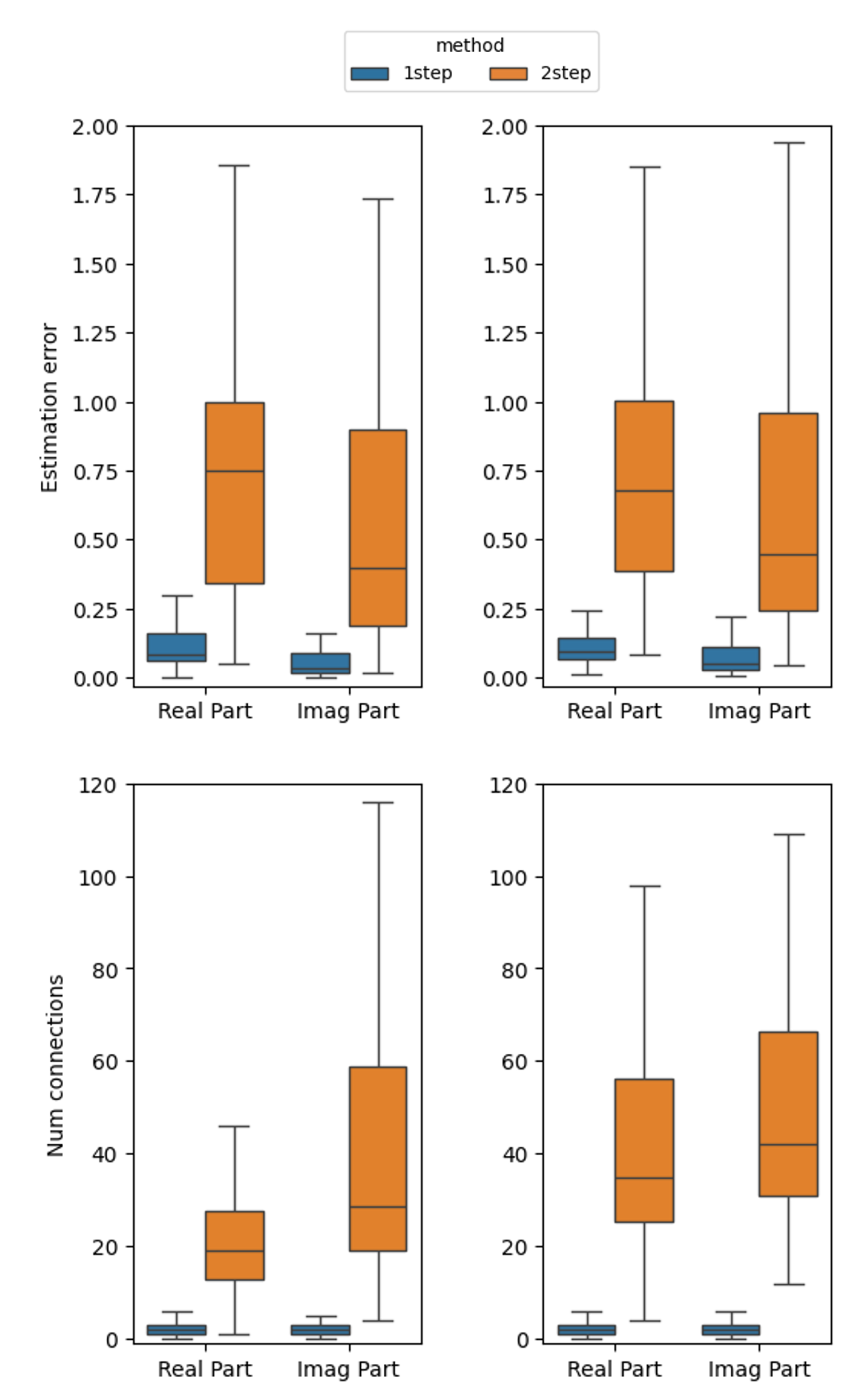}
    \caption{Quantitative comparison of the performance of the one-step and the two-step approach. First row: error in estimating the real and the imaginary part of the cortical cross-power spectrum. Second row: number of supra-threshold connections in the real and imaginary part of the estimated cross-power spectrum. In each row the left-side panel refers to the first simulated brain configuration where only 1 pair of truly connected sources are present, while the right-side panel refers to the second brain configuration involving two pairs of interacting sources. Boxplots summarize results across 50 different simulated data. For the ease of visualization outliers were omitted.}\label{fig:quant_comp}
\end{figure}

\section{Conclusions}\label{sec:conclusions.}

The estimation of  the brain cortical cross-power spectrum from M/EEG data is typically achieved in two steps: first the Tikhonov's least square estimator $\left\{\mbf{x}_\lambda(t)\right\}_{t \in \mathbb{R}}$ of the brain cortical activity is computed, and then the cross-power spectrum of $\left\{\mbf{x}_\lambda(t)\right\}_{t \in \mathbb{R}}$ is estimated through e.g. Welch's method. However this approach often results in a large number of false positives. This issue is partially overcome by deriving from the cross-power spectrum metrics, such as the imaginary part of coherency, that are insensitive to linear mixing of the truly interacting brain source. However, such metrics fail in identifying instantaneously correlated source. In this work we overcome this issue by suggesting a one-step approach that directly estimates the cross-power spectrum of the cortical sources from that of the recorded M/EEG time--series. The number of false positives is reduced by constructed an $\ell_1$--regularized least square estimator computed through FISTA. The proposed method leverages the tensorial structure of the global coefficient matrix. Therefore,  
we developed an efficient algorithm where the computational cost is primarily driven by operations involving the forward matrix $\mbf{G}$.
The advantages of our method over the classical two-step approach have been demonstrated on a large number of synthetic data mimicking different brain configurations.

From a methodological point of view, two are the main developments we are planning. The first one concerns the implementation of an automatic procedure to choose the regularization parameter, possibly employing a back-tracking strategies \cite{scheinberg2014}. The second one concerns implementing different approaches for solving the inverse problem in Eq. (\ref{eq:fwd_vec}), such as Bayesian Monte Carlo approaches \cite{sommariva2014}. The latter would have the advantage of also quantifying the uncertainty of the provided estimation even though at the price of an higher computational cost.

Finally, we observe that the current implementation of the proposed approach provide estimates of the cross-power spectrum at fixed frequencies. Future work will be devoted to investigate how to combine the information from multiple frequencies and/or frequencies ranges.
% This section is a summary of the major results of the paper, without formulas.
\begin{acknowledgements}
The authors are thankful to Dr. Martina Amerighi and Dr. Elisabetta Vallarino for their insightful discussions and suggestions.

L.C., I.F., and S.S. are member of “Gruppo Nazionale per il Calcolo Scientifico" (INdAM-GNCS) and their work is partially supported by INdAM - GNCS Project “Analisi e applicazioni di matrici strutturate (a blocchi)"  CUP E53C23001670001.

The work of I.F. is supported by $\#$NEXTGENERATIONEU (NGEU) and funded by the Ministry of University and Research (MUR), National Recovery and Resilience Plan (NRRP), project MNESYS (PE0000006) – A Multiscale integrated approach to the study of the nervous system in health and disease (DN. 1553 11.10.2022) for the part concerning the analysis of system structure for computation optimization, and its application to
brain connectivity.

S.S. acknowledges the support of the PRIN PNRR 2022 Project 'Computational mEthods for  Medical Imaging (CEMI)' 2022FHCNY3, cup: D53D23005830006 for the conceptualization of the core optimization method in the context of multivariate statistical analysis.\\

\textbf{Data Availability}
The codes use for generating and analysing the datasets used  during the current study are available in the GitHub repository, \url{https://github.com/theMIDAgroup/fista_cps_conn}
\end{acknowledgements}

\end{document}